\documentclass[11pt]{article}

\usepackage{amsmath,graphicx}
\usepackage{tabularx}
\usepackage{amsfonts}
\usepackage{stackrel}
\usepackage{amsthm}
\usepackage{bbm,fullpage}
\usepackage{bm}
\usepackage{thmtools} 

\usepackage{newpxtext}
\usepackage[varg,bigdelims]{newpxmath}

\usepackage[usenames,dvipsnames]{xcolor}

\usepackage{tcolorbox,enumerate}
\tcbuselibrary{skins,breakable}
\tcbset{enhanced jigsaw}

\colorlet{YellowOrange}{RawSienna}

\usepackage[backend=biber, style=alphabetic, backref=true,maxbibnames=99, url=false]{biblatex} 
\usepackage{bbm}
\usepackage{hyperref}
\hypersetup{
    colorlinks=true,
    linkcolor=violet,
    filecolor=magenta,      
    urlcolor=cyan,
    citecolor=blue,
    pdffitwindow=true,
}\usepackage[capitalize, nameinlink]{cleveref}

\usepackage{algorithm,algpseudocode}
\usepackage{multicol}
\usepackage[scaled]{helvet} 

\theoremstyle{plain}
\newtheorem{theorem}{Theorem}[section]
\newtheorem{lemma}[theorem]{Lemma}
\newtheorem{corollary}[theorem]{Corollary}

\newtheorem{fact}[theorem]{Fact}
\newtheorem{conjecture}[theorem]{Conjecture}

\newtheorem{claim}[theorem]{Claim}

\newtheorem{definition}[theorem]{Definition}

\theoremstyle{remark}

\def\eps{\epsilon}
\renewcommand{\P}[1]{\mathbb{P}\left[#1\right]}

\newcommand{\PP}[2]{\mathbb{P}_{#1} \left[#2\right]}
\newcommand{\E}[1]{\mathbb{E}\left[#1\right]}

\newcommand{\seteq}{\mathrel{\mathop:}=}

\newcommand{\R}{\mathbb R}

\usepackage{authblk}

\title{Counting and Sampling Perfect Matchings\\ in Regular Expanding Non-Bipartite Graphs}

\DeclareMathOperator{\poly}{poly}

\author{Farzam Ebrahimnejad\thanks{\href{mailto:febrahim@cs.washington.edu}{febrahim@cs.washington.edu}
Research supported in part by NFS grant CCF-1552097.}}
\author{Ansh Nagda\thanks{\href{mailto:ansh@cs.washington.edu}{ansh@cs.washington.edu}}}
\author{Shayan Oveis Gharan\thanks{\href{mailto:shayan@cs.washington.edu}{shayan@cs.washington.edu}. Research supported by Air Force Office of Scientific Research grant FA9550-20-1-0212, NSF grants  CCF-1552097, CCF-1907845, and a Sloan fellowship.}} 
\affil{University of Washington}

\begin{document}
\maketitle

\thispagestyle{empty}

\begin{abstract}
We show that the ratio of the number of near perfect matchings to the number of perfect matchings in  $d$-regular strong expander (non-bipartite) graphs, with $2n$ vertices, is a polynomial in $n$, thus the Jerrum and Sinclair Markov chain \cite{JS89} mixes in polynomial time and generates an (almost) uniformly random perfect matching. Furthermore, we prove that such graphs have at least $\Omega(d)^{n}$ many perfect matchings, thus proving the Lovasz-Plummer conjecture \cite{LP86} for this family of graphs.
\end{abstract}
\clearpage
\setcounter{page}{1}

\section{Introduction}
Given a (general) graph $G=(V,E)$ with $2n=|V|$ vertices, 
the problem of counting the number of perfect matchings in $G$ is one of the most fundamental open problems in the field of counting. Jerrum and Sinclair  in their landmark result \cite{JS89} designed a Monte Carlo Markov Chain (MCMC) algorithm for this task and proved that such an algorithm runs in polynomial time if the ratio of the number of near perfect matchings to the number of perfect matchings is bound by a polynomial (in $n$). As a consequence one would be able to count perfect matchings if $G$ is {\em very} dense, i.e., it has min-degree at least $n$. Not much is known beyond this case, despite several exciting results when the given graph $G$ is bipartite \cite{JV96,LSW00,Bar99,JSV04,BSVV06}.

This problem is also extensively studied in combinatorics. Around 40 years ago, Falikman and Egorychev \cite{Ego81,Fal81} proved the van-der-Waerden conjecture, thus showing that if $G$ is a $d$-regular bipartite graph, then it has at least $(d/e)^{n}$ perfect matchings. This bound was further improved by Schrijver \cite{Sch98} and simpler and more general proofs were found \cite{Gur06, AOV21}. But it remains a mystery whether van-der-Waerden conjecture extends to non-bipartite graphs. Lovasz, Plummer most famously made the following conjecture:
\begin{conjecture}[{\cite[Conjecture 8.1.8]{LP86}}]\label{conj:LP}
For $d\geq 3$, there exist constants $c_1(d), c_2(d)>0$ such that any $d$-regular $k-1$-edge connected graph $G$ with $2n$ vertices contains at least $c_1(d)c_2(d)^{n}$ perfect matchings and $c_2(d)\to \infty$ as $d\to\infty$. 
\end{conjecture}
To this date the above conjecture is only proved for $d=3$ \cite{EKFN11}, although the same proof shows that the conjecture holds for all $d\geq 3$ as long as $c_2(d)$ is allowed to be a fixed constant.

At a high-level, the study of perfect matchings in general graphs faces the following barriers:
\begin{itemize}
\item Unlike bipartite graphs, the perfect matching polytope of a general graph has exponentially many constraints, and it is believed that there does not exist any poly-size convex program to test whether a given graph has perfect matchings. This fact significantly limits exploiting Gurvits' like techniques	 \cite{Gur06} in lower-bounding the number of perfect matchings.
\item In a bipartite graph, any odd alternating {\em walk} (that starts and ends at  un-saturated vertices) can be used to extend a near perfect matching to a perfect matching. However, in a general graph, an odd alternating walk may contain odd cycles. Therefore, typical augmenting path arguments which bound the ratio of near perfect to perfect matchings fail in a non-bipartite graph (see e.g., \cite{JV96}). 
\end{itemize}
In this paper we study perfect matchings in regular {\em strong} expander graphs: We show that for these graphs the classical algorithm of \cite{JS89} runs in polynomial time and can generate an approximately uniform random perfect matching. On the combinatorial side, we prove a significantly stronger version of \cref{conj:LP} for this family of graphs.

\subsection{Main Contributions}
Given a  graph $G = (V, E)$, let $A_G\in \mathbb{R}^{2n\times 2n}$ be its adjacency matrix, and let $D\in \R^{2n\times 2n}$ be the diagonal matrix of vertex degree. The normalized adjacency matrix of $G$ is defined as $\tilde{A}_G=D^{-1/2}AD^{-1/2}$; when $G$ is clear in the context we may drop the subscript. Let $\lambda_1\geq \lambda_2\geq \dots\geq \lambda_{2n}$ be the eigenvalues of $\tilde{A}$. We write 
$$\sigma_2(\tilde{A})=\max\{\lambda_2,|\lambda_{2n}|\},$$
to denote the largest eigenvalue of $\tilde{A}$ in absolute value (excluding $\lambda_1$). 

\begin{definition}
For $0<\epsilon<1$, we write $G$ is an $\epsilon$-spectral expander if $\sigma_2(\tilde{A})\leq \eps$.
\end{definition}


For two probability distributions $\mu,\nu$ defined in $\{1,\dots,n\}$, the total variation distance of $\mu,\nu$ is $\frac12\sum_{i=1}^n |\mu_i-\nu_i|$.
\begin{theorem}[Algorithm]\label{thm:alg}
There is a randomized algorithm that for $\eps\leq 1/11$, $\delta >0$, given a $d$-regular $\epsilon$-spectral expander $G$ on $2n$ vertices outputs a perfect matching of $G$ from a distribution $\mu$ of total variation distance $\delta$ of the uniform distribution (of perfect matchings) in time $\poly(n^{\log_{1/\eps}d} ,\log(1/\delta))$. Furthermore, there is a randomized algorithm that for any $\delta>0$ approximates the number of perfect matchings of $G$ up to $1\pm\delta$ multiplicative factor in time $\poly(n^{\log_{1/\eps}d} ,1/\delta)$.
\end{theorem}
In particular, observe that the running time of the above algorithms is polynomial in $n$ if $d$ is a constant or $1/\eps$ is a polynomial in $d$ and it is  quasi-polynomial in $n$ otherwise.

\begin{theorem}[Lower Bound]\label{thm:lower-bound}
	For any $\eps\leq 1/11$, every $d$-regular $\epsilon$-spectral expander on $2n$ vertices has at least $(d/e)^n \left(\frac{\eps}{2e^3 d^6}\right)^{\eps n}$ many perfect matchings.
\end{theorem}
Putting the above theorem together with \cite{EKFN11} proves \cref{conj:LP} for (strong) spectral  graphs.

Recall that by a work of Friedman, a random $d$-regular graph is a $\eps=\frac{2\sqrt{d-1}+o(1)}{d}$-spectral expander with probability $1-1/\poly(n)$\cite{Fri08,Bor19}. 
So, for a sufficiently large value of $d$, we can  count the number perfect matchings in random $d$-regular graphs up to $1\pm\delta$-multiplicatively in time polynomial in $n,1/\delta$.  Furthermore, the above theorem implies that the  Lovasz-Plummer \cref{conj:LP} holds for almost all graphs.

We remark our proof technique can naturally be extended to non-regular expanders where the ratio of maximum to minimum degree is bounded. However, in the following statement we show that if this ratio is unbounded the graph may not even have a single perfect matching.
\begin{theorem}\label{cor:expnoperfect}
For $d \geq 3$, there exists $n_0 > 0$ such that for any $n \geq n_0$, there is a
 $O(1/\sqrt{d})$-spectral expander $G$ on $2n$ vertices that does not have any  perfect matchings.
\end{theorem}

\subsection{Related Works}

Bollab\'as and McKay \cite{BM86} showed that when  $d = O(\log^{1/3} n)$, as $d \rightarrow \infty$, a random $d$-regular graph on $2n$ vertices contains $\Omega(d)^n$ many perfect matchings with probability $1-O(1/d^2)$. 
Note that \cref{thm:lower-bound} implies that this statement is true with probability $1-1/\poly(n)$ even if $d = \omega(\log^{1/3} n)$.

Chudnovsky and Seymour \cite{CS12} proved that any planar cubic graph with no cut edge has at least $2^{n/655978752}$ many perfect matchings. 
Building on \cite{CS12}, Esperet, Kardos, King, Kr\'al, and Norine \cite{EKFN11} showed that any $d$-regular $d-1$ edge connected graph has at least $2^{(1-3/d)\frac{n}{3656}}$ perfect matchings.
Barvinok \cite{Bar13} showed that any $3$-regular graph in which  any set $S$ with $2\leq |S|\leq |V|-2$ satisfies $|E(S,\overline{S})|\geq 4$ has at least $c^n$ many perfect matchings for some universal constant $c>1$. 

Jerrum and Sinclair \cite{JS89} showed the ratio of perfect to near perfect matchings in {\em bipartite} Erd\"os-R\'eyni  graphs is polynomial in $n$. Thus, one can efficiently sample a perfect matching in such graphs. However, to the best of our knowledge, no such result is known for (non-bipartite) random  graphs.

Barvinok \cite{Bar99} designed a randomized $c^n$ approximation algorithm to the number of perfect matchings of any (general) graph, for some universal constant $c>1$.
Rudelson, Samarodnitsky, Zeitouni \cite{RSZ16} showed that for a family of strong expander graphs Barvinok's estimator \cite{Bar99}  has a sub-exponential variance, thus obtaining a randomized polynomial time sub-exponential approximation algorithm for the number of perfect matchings of any such graphs.

Gamarnik and Katz \cite{GK10} designed a {\em deterministic} $(1+\epsilon)^n$ approximation algorithm to the number of perfect matching in expanding {\em bipartite} graphs.

\subsection{Overview of Approach}
At high-level our proof builds on works of \cite{JV96,GK10}. We show that given a non-perfect matching $M$ in a (strong) expander graph $G$, one can find many augmenting paths of length $O(\log \frac{n}{n-|M|})$. 
\begin{restatable}{lemma}{nearpertoper}\label{lem:main}
Let $G$ be a $d$-regular $\epsilon$-spectral expander graph on $2n$ vertices for $\eps\leq 1/11$, and let $M$ be any (not perfect) matching in $G$. Then there exist at least $\lceil(n-|M|)/2\rceil$ augmenting paths in $G$ of length at most $\rho = O\left(\max\left(\log_{1/\epsilon}(\frac{2\eps n}{n-|M|}),1\right)\right)$ for $\rho$ defined in \cref{lem:ULURmain}.
\end{restatable}
As alluded to in the introduction, the main difficulty in proving the above theorem is that since $G$ is not necessarily bipartite, an augmenting walk cannot necessarily be turned into an augmenting path since it may have odd cycles. To avoid this issue, first we construct a random bi-partitioning of the vertices of $G$ by placing the endpoints of each edge of $M$ on opposite sides. We exploit the expansion property of $G$ to argue that, 
under this random bi-partition, every set expands with high probability. So, one can start from two unsaturated vertices and follow  ``alternating BFS trees'' from each until getting to a common middle point. The expansion property allows us to show that, with high probability, after $\log_{1/\eps} n$ steps we can construct an augmenting path. 
 This method essentially tries to mimic the approach of \cite{JV96} while exploiting the random partitioning.
  As an immediate corollary of the above lemma, we can upper bound the ratio of $k$ to $k+1$ matchings in expanders. 
\begin{restatable}{lemma}{ratio}\label{lem:ratio}
Let $G$ be a $d$-regular $\epsilon$-spectral expander graph on $2n$ vertices, and let $k\in [n]$. Let $m(j)$ denote the number of matchings of size $j$ in $G$. Then we have
\[\frac{m(k)}{m(k+1)}\leq \frac{2(k+1)}{n-k}d^{(\rho-1)/2} 
\]
for $\rho$ defined in \cref{lem:ULURmain}. 
\end{restatable}
Building on \cite{JS89}, this lemma is already enough to prove \cref{thm:alg}.

To prove \cref{thm:lower-bound}, we first show that for some constant $\eps>0$, $G$ has at least $\Omega(d)^n$ many $n(1-\eps)$-matchings. This part uses a greedy algorithm to find so many distinct matchings in an expander graph. Then, we exploit the above lemma to argue that the ratio of the number of $n(1-\eps)$ matchings of $G$ to the number of its perfect matchings is at most $d^{O(\eps) n}$.

\section{Preliminaries}


Given a graph $G = (V, E)$ with $|V| = 2n$ and $k\in [n]$, a $k$-matching $M\subseteq E$ is any subset with $|M| = k$ and $e\cap e' = \emptyset$ for all $e\neq e'\in M$. 
For a set $S\subseteq V$, we write $G[S]$ to denote the {\em induced} subgraph on the set $S$. For a vertex $v\in V$, we write $\deg_G(v)$ to denote the degree of $v$ in $G$.

Given a set of vertices $S\subseteq V$, define
$$ M(S):=\{v: \exists u\in S, (u,v)\in M\}.$$
We also define $m_G(k)$ to denote the number of $k$-matchings in $G$. 

Given a matching $M$, a walk $v_0,v_1,\dots,v_k$ is an \emph{alternating walk} for $M$ if for any $1\leq i\leq k-1$ exactly one of $(v_{i-1},v_i)$ and $(v_i,v_{i+1})$ is in $M$.
An \emph{augmenting path} for $M$ is any alternating path that starts and ends with an unmatched vertex.

For a graph $G=(V,E)$ and $S,T\subseteq V$, 
$$E_G(S, T):=\{(u,v)\in  S\times T : (u,v)\in E\}.$$ 
For a set $S\subseteq V$, we write 
$$N_G(S):=\{u\notin S: \exists u\in S, (u,v)\in E\}$$ 
to denote the set of all vertices outside $S$ that has an edge to $S$. When the graph $G$ is unambiguous from the context, we may drop the subscripts.

\subsection{Spectral Graph Theory}

The following facts are the main properties of spectral expanders that we will need.

\begin{fact}[Expander Mixing Lemma]
\label{fact:expander-mixing-lemma}
Let $G$ be a $d$-regular graph on $2n$ vertices. Then for any two sets $S,T\subseteq V$, we have
$$\Big| |E(S,T)| - \frac{|S|\cdot |T|}{2n}\Big| \leq d \sigma_2(\tilde{A})\sqrt{|S|\cdot |T|}$$
\end{fact}


\begin{lemma}
\label{claim:average-degree}
	Let $G = (V, E)$ be a $2n$-vertex $d$-regular $\epsilon$-expander, and let $S \subseteq V$. The following holds:
	Then, there exists $v \in S$ such that $\deg_{G[S]}(v) \geq \lceil d (|S|/2n-\eps)\rceil$.
\end{lemma}
\begin{proof}
By the Expander Mixing Lemma (\cref{fact:expander-mixing-lemma}), we have $|E(S, S)|\geq \frac{d|S|^2}{2n} - d\epsilon |S|$. 
Hence the average degree of the vertices in $G[S]$ is at least $d (|S|/2n-\eps)$, and in particular there exists $v \in S$ whose degree in $G[S]$ is at least that much.
\end{proof}

\begin{lemma}[\cite{tanner1984explicit}]\label{lemma:tanner}
	Let $G$ be a $d$-regular $\eps$-expander on $2n$ vertices. Then for any $S\subseteq V$ we have
	\[|N(S)|\geq \frac{|S|}{\eps^2 + (1-\eps^2)|S|/2n}\]
\end{lemma}

When $|S|\leq 2\eps n$, we immediately get the following corollary.

\begin{corollary}\label{claim:expander-properties}
	Let $G$ be a $d$-regular $\eps$-expander on $2n$ vertices. Then for any $S\subseteq V$ with $|S|\leq 2\epsilon n$ we have
	\[|N(S)|\geq \frac{|S|}{\eps^2 + \eps -\eps^3}\geq \frac{|S|}{\eps^2 + \eps}.\]

\end{corollary}

\subsection{Inequalities}
\begin{theorem}[Hoeffding's Inequality]\label{thm:hoeffding}
Let $X_1,\dots,X_k$ be independent random variables in the range $[0,1]$. Then,
	$$ \P{\sum X_i < \mathbb{E} \sum X_i -\epsilon} \leq \exp(-2\epsilon^2/k).$$
\end{theorem}
\begin{theorem}[Stirling's Formula]
\label{thm:stirling}
For $n \geq 1$ we have
\[
n! \geq \left(\frac{n}{e}\right)^n.
\]
\end{theorem}


\begin{theorem}[{Weierstrass's Inequality}]\label{thm:weierstrass}
	Let $0<x_i<1$ for $1\le i\le n$. Then,
	$$ \prod_{i=1}^n (1-x_i) \geq 1-\sum_{i=1}^n x_i.$$
\end{theorem}

\begin{theorem}[Hoffman-Wielandt's Inequality]
\label{thm:hoffman}
	Let $A,B \in \R^{n \times n}$ be symmetric matrices with eigenvalues $\lambda_1 \geq \cdots \geq \lambda_n$ and $\lambda'_1 \geq \cdots \geq \lambda'_n$, respectively. We have
	\[
	\sum_{i=1}^n (\lambda_i - \lambda'_i)^2 \leq \|A - B\|_F^2,
	\]
	where $\|\cdot\|_F$ denotes the Frobenius norm.
\end{theorem}


\section{Proof of the Main Lemma}
The following lemma is the main result of this section.
\begin{lemma}\label{lem:ULURmain}
Let $G=(V,E)$ be a $d$-regular $\epsilon$-spectral expander graph on $2n$ vertices with $\epsilon\leq 1/11$,  $M$ be any (not perfect) matching in $G$, and  $U$ the set of unsaturated vertices (in $M$).
For any partitioning of $U=U_L \cup U_R$ with $|U_L|=|U_R|$ there is an augmenting path from $U_L$ to $U_R$ of length at most 
$\rho = 4\max\left(\lceil \log_{C_1(\eps)} (\frac{2\epsilon n+1}{n-|M|}) \rceil, 0\right) + 1$, where $C_1(\eps)=\frac{1}{\eps+\epsilon^2}$.  
\end{lemma}
Before proving this lemma we use it to prove \cref{lem:main}.
\begin{proof}[Proof of \cref{lem:main}]	
Let $U$ be the set of unmatched vertices of $M$ and let $U'$ be 
vertices of $U$ that are not an endpoint to any augmenting path of length at most $\rho$.
Observe that if $|U'|<=n-|M|$, then there are at least $\lceil(n-|M|)/2\rceil$ augmenting paths for $M$ and we are done.

For the sake of contradiction suppose $|U'|>n-|M|$. 
Now arbitrarily partition $U$ into two equal-sized sets $U_L\cup U_R$ (each of size exactly $n-|M|$) with the constraint that $U_L\subseteq U'$. So, by construction, no vertex in $U_L$ is an endpoint of augmenting path of length at most $\rho$.
But, by \cref{lem:ULURmain} there is an augmenting path from $U_L$ to $U_R$ of length at most $\rho$ which is a contradiction.
\end{proof}
\begin{proof}[Proof of \cref{lem:ratio}]
Given a $k$-matching $M$, by \cref{lem:main} there are at least $(n-k)/2$ augmenting paths for $M$ (in $G$) of length at most $\rho$ for $\rho$ defined in \cref{lem:main}.
Note that for any vertex $v$ of $G$ the number of paths of length at most $\rho$ starting at $v$ is at most $d^\rho$.
Therefore, for any $k+1$-matching $M'$, there are at most $2(k+1)d^{(\rho-1)/2}$ $k$-matchings  that can be mapped to $M'$. 
This is because any such matching can be obtained by ``undoing'' an alternating path that starts and ends at the saturated vertices of $M'$.
Together, these imply 
$\frac{m(k)}{m(k+1)}\leq \frac{2(k+1)d^{(\rho-1)/2}}{n-k}$.
\end{proof}


\begin{definition}[Bipartition of $G$]\label{def:bipartition} Given a matching $M$ and $\omega:M\to\{0,1\}$, we define the bipartite graph $G_M(\omega)=(L_M(\omega),R_M(\omega),E_M(\omega))$ as follows. We drop the subscript $M$ and $\omega$ if they are clear in the context.
 
All vertices of $U_L$ are in $L$, all vertices of $U_R$ are in $R$. For any edge $e=(u,v)\in M$, we add $u$ to $L$ and $v$ to $R$ if $\omega(e)=0$ and we add $u$ to $R$ and $v$ to $L$ otherwise.
We simply let $E_M(\omega)$ be all edges of $E$ connecting $L$ to $R$.
We use $\mu_M$ to denote the uniform distribution over functions $M\to\{0,1\}$.
\end{definition}



\begin{lemma}
\label{lem:short-alternating-path}
Let $G=(V,E)$ be a graph with $2n$ vertices 
such that for every set $S\subseteq V$ with $|S|\leq 2\eps n$, $|N(S)|\geq \alpha |S|$ for $\alpha\geq 10$ and $0<\eps<1$. Given a non-perfect matching $M$ and a partition of non-saturated vertices into equal sized sets $U_L,U_R$,
if for $t= \max(\lceil\log_{\alpha/4}{\frac{2\eps n+1}{|U_L|}}\rceil, 0)$ there is no augmenting path of length at most $4t+1$ from $U_L$ to $U_R$, then with probability  $>1/2$ (for $\omega\sim\mu_M$) there exists a set $S\subseteq L$ such that $|S|> 2\eps n$, and for every $v\in S$ there is an alternating path of length at most $2t$ from $U_L$ to $v$ in $G_M(\omega)$.
\end{lemma}
\begin{proof}[Proof of \cref{lem:ULURmain}] 
First, by \cref{claim:expander-properties}, since $G$ is an $\eps$-spectral expander and $\eps<1/11$, we can let $\alpha=1/(\eps + \eps^2) \geq 10$.
We prove the claim by contradiction. Suppose $G$ has no augmenting path of length $\rho\seteq 4t+1$ from $U_L$ to $U_R$, for $t$ defined in \cref{lem:short-alternating-path}.
By \cref{lem:short-alternating-path}, for $\omega\sim\mu$, with probability $>1/2$ there is a set $S\subseteq L$ with $|S|>2\eps n$, such that for any $v\in S$ there is an alternating path in $G_M(\omega)$ of length (at most) $2t$ from $U_L$ to $v$.
By renaming $U_L,U_R$, with probability  $>1/2$ there also exists another set $S'\subseteq R$ such that $|S'|> 2\epsilon n$ such that for every $v\in S' $, there is an alternating path of length at most $2t$ from $U_R$ to $v$ in $G_M(\omega)$.
By union bound, with positive probability  both of these sets exist.
Now, by \cref{fact:expander-mixing-lemma} we have
\[
|E(S, S')|\geq \frac{d|S|\cdot |S'|}{2n} - \epsilon d\sqrt{|S|\cdot |S'|} > d\sqrt{|S|\cdot |S'|}(\epsilon - \epsilon) = 0.
\]
So there is an edge $(v, v')\in E(S,S')$. Now, the path formed by concatenating an alternating path from $U_L$ to $v$ of length $2t$, the edge $(v,v')$, and an alternating path from $v'$ to $U_R$ of length $2t$  we find alternating walk of length (at most) $\rho=2t + 2t + 1$ from $U_L$ to $U_R$ in $G_M(\omega)$. But since $G_M(\omega)$ is a bipartite graph this walk can only have even length cycles; by removing these cycles we obtain an alternating path of length at most $\rho$ from $U_L$ to $U_R$ (in $G_M(\omega)$).
\end{proof}

In the rest of this section, we prove \cref{lem:short-alternating-path}. First note that as $\alpha / 4 >1$, we have $\log_{\alpha/4}{\frac{2\eps n+1}{|U_L|}} \leq 0$ if and only if $|U_L| > 2\epsilon n$, and the claim is trivial in this case as we can set $S = U_L$. 
Now suppose $|U_L| \leq 2\epsilon n$. Let us fix an arbitrary ordering on the vertices of $G$. Given a bipartition $G_M(\omega)$, we define a sequence of sets $U_L=L_0\subseteq L_1\subseteq \dots L_T \subseteq L$, and $\emptyset= X_0\subseteq X_1\subseteq \dots X_T\subseteq V$, where $T$ is a stopping time which is the minimum of $t$ and the first time that $|L_T|>\eps n$. 
 Given $L_{i-1}, X_{i-1}$ for $i\geq 1$, we construct $L_i,X_i$ as follows: If $|L_{i-1}|>\eps n$ then we stop and we let $T=i-1$. Otherwise, $|L_{i-1}|\leq 2\eps n$ so by assumption of the lemma, $N(L_{i-1})\geq \alpha|L_{i-1}|$. Let $A_i$ be the 
lexicographically first $\alpha |L_{i-1}| - |X_{i-1}|$ neighbors of $L_{i-1}$ which are not in $X_{i-1}$. In other words, we sort all neighbors of $L_{i-1}$ which are not in $X_{i-1}$ lexicographically and we let the first $\alpha |L_{i-1}|-|X_{i-1}|$ of them to be $A_i$.
Note that as $L_{i-1}$ has at least $\alpha|L_{i-1}|$ neighbors, there are at least $\alpha |L_{i-1}| - |X_{i-1}|$ ``new'' neighbors and so the set $A_i$ is well-defined. We let $X_i=X_{i-1}\cup A_i$. Observe that by definition, we always have
\begin{equation}\label{eq:Xisize} |X_i|=\alpha|L_{i-1}|.	
\end{equation}
Finally, we let
$$L_i = L_{i-1} \cup M(A_i\cap R) = L_{i-1}\cup M(X_i\cap R).$$
The following fact follows inductively from the above construction
\begin{fact}
For every $1\leq i\leq T$ and  every $v\in L_i$, there is an alternating path of length at most $2i$ from $U_L=L_0$ to $v$ in $G_M(\omega)$.
\end{fact}



\begin{fact}\label{fact:Li-1Aidisj}
For any $1 \leq i \leq T$, $L_{i-1}\cap M(A_i\cap R)=\emptyset$. Therefore,
$$ |L_i| = |L_{i-1}| + |A_i\cap R|.$$	
\end{fact}
\begin{proof}
	For the sake of contradiction let $v\in M(A_i\cap R)$ such that $v\in L_{i-1}$ as well. Then, since $v$ has a match, $v\notin U_L$; so we must have $i\geq 2$. Let $1\leq j\leq i-1$ be the smallest index such that $v\in L_j$. That means that,  by construction, $M(v)\in A_j\cap R$. Therefore, $M(v)\in X_j\subseteq X_{i-1}$. So, $v\notin A_i$.
\end{proof}
Since in the above construction we only ``look at'' the first $\alpha|L_{i-1}|-|X_{i-1}|$ new neighbors of $L_{i-1}$ to construct $L_i$, it follows that all edges which have no endpoints in these sets  are conditionally independent. More precisely, we obtain the following Fact.
\begin{fact}
\label{claim:independence}
Let $\omega$ be chosen uniformly at random. For any $1\leq i<t$, conditioned on $L_0,\dots,L_{i-1}$, the law of $\omega$ on all edges that have no endpoints in $L_{i-1},X_{i-1}$ remain invariant, i.e., it is i.i.d., with expectation $1/2$ on each edge.
\end{fact}


\begin{claim}
\label{claim:chernoff}
	For  $1\leq i\leq T$, 
	$$\PP{\omega\sim\mu}{|A_i \cap R| \leq  |A_i| / 4 \mid L_0,\dots,L_{i-1}}<\exp(-|A_i|/8).$$
\end{claim}
\begin{proof}
Note that given $L_0,\dots,L_{i-1}$, $X_1,\dots,X_i$ and $A_1,\dots,A_i$ are uniquely determined. 
Let $v\in A_i$. Consider the following cases:
\begin{itemize}
\item $v\in U_R$. This case cannot happen because we get an augmenting path of length $2i+1$ to $U_L$ which is a contradiction.
\item $v\in L_{i-1}$. This cannot happen because $v\in N(L_{i-1})$. This in particular shows $v\notin U_L$. So, $v$ has a match in $M$.
\item $M(v)\in A_i$. Then, by  \cref{def:bipartition} exactly one of $v,M(v)$ is in $R$.
\item $M(v)\in X_{i-1}$. If $M(v)\in R$ then we must have $v\in L_{i-1}$ which cannot happen as we said in case (2). Otherwise, $M(v)\in L$, so $v\in R$.
\item $M(v)\in L_{i-1}$. Then, $v\in R$.
\item $v\notin L_{i-1}, M(v)\notin X_i, L_{i-1}$. In this case since $v\in A_i$, by \cref{claim:independence}, $v\in R$ with probability $1/2$ independent of all other vertices of $A_i$.
\end{itemize}
Let $A'_i$ be the set of vertices $v$ that fall into the last case.
Say we have a Bernoulli $B_v$ with success probability $1/2$ for every $v\in A'_i$. Then, by above discussion, conditioned on $L_0,\dots,L_{i-1}$, with probability 1,
$$ |A_i \cap R|\geq |A_i\setminus A'_i|/2 + \sum_{v\in A'_i} B_v.$$
Therefore, by the Hoeffding bound (\cref{thm:hoeffding})
\begin{align*}\P{|A_i\cap R|\geq |A_i/4| \mid L_0,\dots,L_{i-1}} &\leq \P{\sum_{v\in A'_i} B_v \leq  |A'_i|/2 - |A_i|/4 \Bigm\vert L_0,\dots,L_{i-1}} \\
&\leq \exp(-|A_i|^2/8|A'_i|) \leq \exp(-|A_i|/8)	
\end{align*}
as desired.
%
%
%
\end{proof}

%
Since $|L_1|\geq |A_1\cap R|$ and $|A_1|=\alpha|L_0|$,
\begin{align}\PP{\mu}{|L_1|\geq (\alpha/4)|L_0|} &\geq \PP{\mu}{|A_1\cap R|\geq (\alpha/4)|L_0|} \nonumber\\
	&= \PP{\mu}{|A_1\cap R|\geq |A_1|/4}\geq 1-\exp(-|A_1|/8)=1-\exp(-\alpha|L_0|/8) \label{eq:basecaseL1}
\end{align}
where the last inequality follows form \cref{claim:chernoff}.

\begin{claim}
\label{claim:expanding-Li}
Let $\omega$ be chosen uniformly at random. For every $2\leq i\leq T$, we have
	\[
	\PP{\mu}{|L_{i}| \geq (\alpha / 4) |L_{i-1}| \bigm\vert |L_{i-1}| \geq (\alpha / 4) |L_{i-2}|, L_0,\dots,L_{i-1}} \geq 1 - \exp(-(\alpha - 4) |L_{i-1}|/8).
	\]
\end{claim}
\begin{proof}
Suppose $|L_{i-1}| \geq (\alpha / 4) |L_{i-1}|$. Recall that by \cref{eq:Xisize} we have $|X_i|=\alpha|L_{i-1}|$. So we can write
\begin{align}
\begin{split}
\label{eq:Ai_Li}
	|A_i| = |X_i \setminus X_{i-1}| &= \alpha(|L_{i-1}| - |L_{i-2}|) \\
	&\geq \alpha  (1 - 4/\alpha)  |L_{i-1}| \\
	&= (\alpha - 4) |L_{i-1}|.
\end{split}
\end{align}

 Let $\mu'$ be $\mu$ conditioned on $|L_{i-1}|\geq (\alpha/4)|L_{i-2}|$ and $L_0,\dots,L_{i-1}$. Then,
\begin{align*}
	\PP{\mu'}{|L_i| \leq (\alpha / 4) |L_{i-1}|} 
	&= \PP{\mu'}{|L_i| - |L_{i-1}| \leq (\alpha / 4 - 1) |L_{i-1}|}  \\
	&= \PP{\mu'}{|A_i \cap R| \leq (\alpha / 4 - 1) |L_{i-1}|} \tag{\cref{fact:Li-1Aidisj}}\\
	&\leq \PP{\mu'}{|A_i \cap R| \leq |A_i| / 4} \tag{\cref{eq:Ai_Li}}\\
	&\leq \exp(-|A_i| / 8)\tag{\cref{claim:chernoff}}\\
	&\leq \exp(-(\alpha - 4) |L_{i-1}|/8) \tag{\cref{eq:Ai_Li}}, 
\end{align*}
completing the proof. 
\end{proof}

\begin{claim}
If $\alpha\geq 10$, for any $i\geq 1$ we have
	\[
	\PP{\mu}{T<i \vee (T\geq i \wedge |L_i| \geq (\alpha/4)^i |L_0|)} >1/2.
	\]
\end{claim}
\begin{proof}
For $1\leq j\leq i$, let $E_j$ denote the event that $T<j\vee (T\geq j \wedge |L_j| \geq (\alpha / 4) |L_{j-1}|)$. 
Then,
\begin{align*}
\P{E_i} & \geq \P{E_1\wedge \dots \wedge E_i} \\
&= \P{E_1} \prod_{j=2}^i \P{E_j | E_1,\dots,E_{j-1}}\\
&\geq (1-\exp(-\alpha|L_0|/8)) \prod_{j=2}^i \Big(\P{T<j|E_1,\dots,E_{j-1}}\\
&\quad+\P{T\geq j|E_1,\dots,E_{j-1}}\E{1-\exp\left(\frac{-(\alpha-4)|L_{j-1}|}{8}\right) \Bigm\vert E_1,\dots,E_{j-1},T\geq j}\Big) \tag{\cref{claim:expanding-Li}}\\
&\geq (1-\exp(-\alpha|L_0|/8)) \prod_{j=2}^i \left(1-\exp\left(\frac{-(\alpha-4)}{8}(\alpha/4)^{j-1}|L_0|\right)\right) \\
&\geq 1-\exp(-\alpha/8)-\sum_{j=2}^i \exp(-\frac{\alpha-4}{8}(\alpha/4)^{j-1}) \tag{\cref{thm:weierstrass}}\\
&\geq 1-e^{-\alpha/8}-\sum_{j=0}^{\infty} e^{-\beta (\alpha/4)^j}\tag{for $\beta=\frac{\alpha(\alpha-4)}{32}$}\\
&\geq 1-e^{-\alpha/8}-\frac{e^{-\beta}}{1-e^{-\beta(\alpha/4-1)}}>1/2 \tag{$\alpha\geq 10$}.
\end{align*}
Note that in the third inequality we crucially use that if $E_1,\dots,E_{j-1}$ occur then either $T<j$, or $T\geq j$ and $|L_j|\geq (\alpha/4)^{j-1}$.
\end{proof}
Setting $t = \lceil \log_{\alpha / 4} \frac{2\epsilon n+1}{|U_L|} \rceil$ by the above statement we get $\P{T<t \vee (T=t\wedge |L_t|>2\eps n)}>1/2$. This completes the proof of \cref{lem:short-alternating-path}.	

\section{Completing the proofs of \cref{thm:alg,thm:lower-bound}}

\subsection{The Lower-Bound}

%

\begin{lemma}
\label{lem:Meps-lower-bound}
	Let $G = (V,E)$ be an $2n$-vertex $d$-regular $\eps$-expander. If $\eps<1/2$, then we have,
	\[
	m((1-\epsilon)n)\geq \left(\frac{d}{e}\right)^{n(1-\eps)} \cdot e^{-2\eps n}.
	\]
	\
\end{lemma}
\begin{proof}Let $k= n (1- \eps)$. We call a sequence of integers $\langle a_1, \ldots, a_k \rangle$ \emph{valid} if 
$1 \leq a_i \leq  \lceil d((n-i+1)/n-\eps)\rceil$ for all $1\leq i\leq k$.


Now, for  valid sequence $a = \langle a_1, \ldots, a_k \rangle$, we construct a $k$-matching $\mathcal{M}(a)$ as follows: We are going to construct a sequence of matchings $M_0 \subseteq M_1 \subseteq \cdots \subseteq M_k$, with the property that for $1 \leq i \leq k$, $M_i$ is going to be a matching of size $i$ in $G$. We then set $\mathcal{M}(a) \seteq M_k$.
We start with $M_0 = \emptyset$. For $i \geq 1$, given $M_{i-1}$, let $S_i$ be the set of unmatched vertices of $G$ with respect to $M_{i-1}$. Note that by construction $M_{i-1}$ is a matching of size $i-1$, so we have $|S_i| = 2n - 2(i-1)$. Further let $\Delta_i = \max_{u \in S} \deg_{G[S_i]}(u)$, and and let $u_i$ denote the lexicographically first vertex with degree $\Delta_i$ in $S_i$. Note that by \cref{claim:average-degree}, it should be that $\Delta_i \geq \lceil d((n-i+1)/n-\eps)\rceil$, and furthermore, by validity of $a$ we obtain $a_i  \leq \Delta_i$. Now let 
let $v_i$ be the $a_i$-th neighbor of $u_i$ in $G[S_i]$ with respect to the lexicographical order.  We set $M_i = M_{i-1} \cup \{(u_i, v_i)\}$.
In the next claim we show that any distinct pair of valid sequences give distinct $k$-matchings. Therefore, the number of $k$-matchings of $G$ is at least, 
  	\begin{align*}
  		\prod_{i = 1}^{k} (d\frac{n-i+1}{n}-d\eps)\geq d^k \prod_{i=1}^{k} \frac{n-i+1-\eps n}{n} 
  		\geq d^k \frac{k!}{n^k} 
		\geq (d/e)^k  \cdot (k/n)^k,
  	\end{align*}
  	where the last inequality uses \cref{thm:stirling}. 
  	By plugging in $k = (1-\epsilon) n $ we obtain
  	\begin{align*}
  			m((1-\eps) n) \geq (d/e)^{(1-\eps) n} \cdot (1-\epsilon)^{(1-\eps) n} 
  			\geq (d/e)^{(1-\eps) n} \cdot e^{-\eps n},
  	\end{align*}
  	where in the last inequality we used that $(1 - \epsilon)^{1-\epsilon} \geq e^{-\epsilon}$ for $\epsilon \leq 1/2$.
\end{proof}

\begin{claim}
	For any distinct valid sequences $a = \langle a_1, \ldots, a_k \rangle$ and $b = \langle b_1, \dots, b_k \rangle$ we have $\mathcal{M}(a) \neq \mathcal{M}(b)$.
\end{claim}
\begin{proof}
Since $a\neq b$ there is an index $1\leq i\leq k$ such that $a_i\neq b_i$;  let $1 \leq i \leq k$ be the first such index. Since $a_j=b_j$ for $1\leq j\leq i-1$, by the above construction we have $S_i(a) = S_i(b)$. So, we would choose a unique vertex $u_i$ in both constructions but we match it to different vertices, since $a_i\neq b_i$. Therefore ${\cal M}(a)\neq {\cal M}(b)$.
\end{proof}

%
%
%
%
\begin{lemma}
\label{lem:Mn-Meps-lower-bound}
	Let $G$ be a $2n$ vertex $d$-regular, $\eps$-spectral expander for $\eps\leq 1/11$. We have,
	\[
	\frac{m((1-\epsilon)n )}{m(n)}\leq (2e/\eps)^{\eps n} d^{2\eps n+(4\eps n+2)/\ln C_1(\eps)}.
	\]
	where $C_1(\eps)$ is defined in \cref{lem:ULURmain}.
\end{lemma}
\begin{proof}
For $k=\eps n$, we can write
\begin{align*}
	\frac{m(n(1-\eps))}{m(n)} =  \prod_{i=n-k}^{n-1} \frac{m(i)}{m(i+1)}
	&\leq  \prod_{i=n-k}^{n-1} \frac{2(i+1)}{n-i} d^{2\log_{C_1(\eps)}\frac{2\eps n+1}{n-i}+2} \tag{\cref{lem:ratio}}  \\
	&\leq \frac{2^k n^k d^{2k}}{k!} d^{2\sum_{i=1}^{\eps n} \log_{C_1(\eps)}\frac{2\eps n+1}{i}} \\
	&\leq (2ed^2n/k)^k d^{2\log_{C_1(\eps)}\frac{(2k+1)^k}{k!}} \leq (2ed^2n/k)^k d^{(4k+2)/\ln C_1(\eps)},
	\end{align*}
	where in the second to last inequality we used \cref{thm:stirling} and in the last inequality we used $\frac{(2k+1)^k}{k!}=\frac{k^k}{k!}(2+1/k)^k\leq e^{2k+1}$.
	Plugging $k = \epsilon n$ into  the above inequality proves the claim.
\end{proof}
\begin{proof}[Proof of \cref{thm:lower-bound}]
Using \cref{lem:Meps-lower-bound,lem:Mn-Meps-lower-bound} we can write and using  for $\eps\leq 1/11$, $\ln (C_1(\eps))\geq 2$,
\begin{align*}
m(n)\geq \frac{e^{-\epsilon n}(d/e)^{n(1-\eps)}  }{(2e/\eps)^{\eps n}d^{2\eps n+(4\eps n+2)/\ln C_1(\eps)}}	\geq \left(\frac{d}{e}\right)^n \left(\frac{\eps}{2e^3 d^6}\right)^{\eps n}
\end{align*}
as desired. 
\end{proof}

\subsection{Sampling / Counting Perfect Matchings}
As an immediate corollary of \cref{lem:ratio} we prove \cref{thm:alg}.
In particular,
\begin{equation}\label{eq:mn-1mn}
	\frac{m(n-1)}{m(n)} \leq 2n d^{2\log_{C_1(\eps)}(2\eps n+1)+2}
\end{equation}
So, it follows from the following theorem of \cite{JS89} that for any $\delta>0$ we can sample a perfect matching of $G$ from a distribution $\mu$ of total variation distance $\delta$ of the uniform distribution in time $\poly(n^{\frac{\log d}{\log \eps^{-1}}},\log(1/\delta))$.
\begin{theorem}[{Jerrum and Sinclair \cite[Thm 3.6]{JS89}}]
	Let $G$ be a graph with $2n$ vertices. There is a Markov chain with a uniform stationary distribution on the space $n$ and $n-1$ matchings of $G$ such that that mixes in time $\poly(n, \frac{m(n-1)}{m(n)})$.
\end{theorem}

Furthermore, Jerrum and Sinclair \cite[Thm 5.3]{JS89} showed how to estimate the number of perfect matchings up to $1\pm\delta$ multiplicative factor in time $\poly(n,1/\delta,\frac{m(n-1)}{m(n)})$. So, plugging in \cref{eq:mn-1mn} into their theorem also allows us to approximate the number of perfect matchings (up to $1\pm\delta$ multiplicatively) in  $\eps$-expander regular graphs in time $\poly(n^{\frac{\log d}{\log \eps^{-1}}},1/\delta)$.

\section{A non-regular counter-example}
In this section we construct an infinite family of {\em non-regular} strong spectral expanders that do not have any perfect matchings.
This shows that the regularity assumption in \cref{thm:lower-bound} is necessary.
 

  \begin{lemma}
  \label{lem:augmented}
  Given  a $d$-regular graph $G = (V, E)$  with $2n$ vertices, 
  there exists a graph $H = (V', E')$ with $2n + 2$ vertices such that 
  \begin{itemize}
    \item $H$ does not have any perfect matchings.
    \item $\sigma_2(\tilde{A}_H) \leq \sigma_2(\tilde{A}_G) + \sqrt{5/d}$.
  	\item $H$ has $2n-1$ vertices of degree $d$, one vertex of degree $d+2$, and two vertices of degree $1$. 
  \end{itemize}
 \end{lemma}
\begin{proof}
Say $V=\{v_1,\dots,v_{2n}\}$. To construct $H$, we add two new vertices $v_{2n+1},v_{2n+2}$ and we  connect both of them to $v_{2n}$. 
Clearly  $H$ has no perfect matchings. 
We abuse notation and extend the normalize adjacency matrix of $G$, $\tilde{A}_G$ by adding two all-zeros rows and two all-zeros columns. Clearly, only introduces two new zero eigenvalues, and the $\sigma_2(\tilde{A}_G)$ remains invariant. 
It follows by a simple calculation that
\[
\|\tilde{A}_G - \tilde{A}_H\|_F^2 = 2(d-1)\cdot  \left(\frac{1}{d(d+1)}\right)^2 + 4 \left(\frac{1}{\sqrt{d+1}}\right)^2 \leq \frac{2}{d^3}+ \frac{4}{d} \leq \frac{5}{d}.
\]

Therefore, by \cref{thm:hoffman}, for any $1 \leq i \leq 2n+2$ we have 
$$|\lambda_i(\tilde{A}_G) - \lambda_i(\tilde{A}_H)|^2 \leq \sum_{j=1}^{2n+2} |\lambda_j(\tilde{A}_G)-\lambda_j(\tilde{A}_H)|^2 \leq \|\tilde{A}_G - \tilde{A}_H\|^2_F \leq 5/d.$$ 
Therefore we obtain
 $\sigma_2(\tilde{A}_H)\leq  \sigma_2(\hat{A}_G)+\sqrt{5/d}$. 
\end{proof}
Recall by the work Friedman \cite{Fri08,Bor19} for $d \geq 3$ and sufficiently large $n$, there exists a $d$-regular $\left(\frac{2\sqrt{d-1}}{d}+o(1)\right)$-expander $G_{2n,d}$ on $2n$ vertices. \cref{cor:expnoperfect} is immediate.


\printbibliography

\end{document}